\documentclass[10pt,conference]{IEEEtran}
\usepackage{bm}
\usepackage{bbm}
\usepackage{mathrsfs}
\usepackage{textcomp}
\usepackage{epsfig}
\usepackage{indentfirst}
\usepackage{amsmath}
\usepackage{amssymb}
\usepackage{array}
\usepackage{multirow}
\usepackage{graphicx}
\usepackage{subfigure}
\usepackage{multicol}

\usepackage{algorithm}
\usepackage{algpseudocode}

\usepackage{graphicx}
\usepackage{epstopdf}
\usepackage{color}

\usepackage{amsthm}

\usepackage{xcolor}
\usepackage{soul}

\makeatletter
\newcommand\fs@spaceruled{\def\@fs@cfont{\bfseries}\let\@fs@capt\floatc@ruled
  \def\@fs@pre{\vspace{0.8em}\hrule height.8pt depth0pt \kern2pt}%
  \def\@fs@post{\kern2pt\hrule\relax}%
  \def\@fs@mid{\kern2pt\hrule\kern2pt}%
  \let\@fs@iftopcapt\iftrue}
\makeatother

\theoremstyle{plain}

\newcommand{\trace}{\mathop{\mathrm{tr}}}
\newcommand{\rank}{\mathop{\mathrm{rank}}}
\newcommand{\diag}{\mathop{\mathrm{diag}}}

\newtheorem{lemma}{Lemma}

\theoremstyle{definition}

\theoremstyle{remark}
\newtheorem{remark}{Remark}

\usepackage{cite}

\begin{document}
\title{Beamforming Design for Max-Min Fair SWIPT in Green Cloud-RAN with Wireless Fronthaul}
\author{{Zhao Chen$^{1}$, Haisheng Xu$^{1}$, Lin X. Cai$^{2}$, and Yu Cheng$^{2}$}\\
$^{1}$Department of Electrical Engineering, Columbia University, New York, USA\\
$^{2}$Department of Electrical and Computer Engineering, Illinois Institute of Technology, Chicago, USA\\
Emails: \{zc2412, hx2219\}@columbia.edu, \{lincai, cheng\}@iit.edu}

\maketitle

\begin{abstract}
In this paper, a joint beamforming design for max-min fair simultaneous wireless information and power transfer (SWIPT) is investigated in a green cloud radio access network (Cloud-RAN) with millimeter wave (mmWave) wireless fronthaul.  
To achieve a balanced user experience for separately located data receivers (DRs) and energy receivers (ERs) in the network, joint transmit beamforming vectors are optimized to maximize the minimum data rate among all the DRs, while satisfying each ER with sufficient RF energy at the same time. 
Then, a two-step iterative algorithm is proposed to solve the original non-convex optimization problem with the fronthaul capacity constraint in an $l_0$-norm form. 
Specifically, the $l_0$-norm constraint can be approximated by the reweighted $l_1$-norm, from which the optimal max-min data rate and the corresponding joint beamforming vector can be derived via semidefinite relaxation (SDR) and bi-section search. 
Finally, extensive numerical simulations are performed to verify the superiority of the proposed joint beamforming design to other separate beamforming strategies. 

\end{abstract}


\section{Introduction}


With the rapidly increasing demand of data traffic in future wireless communication networks, cloud radio access network (Cloud-RAN)~\cite{chih2014toward,chen2017energy} becomes an emerging network
architecture to achieve high-speed and ubiquitous connectivity with guaranteed quality of service (QoS) in a cost-effective way. 
In a Cloud-RAN, instead of conventional base stations (BSs), low-power and low-complexity remote radio heads (RRHs) are densely deployed and connected via fronthaul links to a pool of baseband processing units (BBUs) at the central processor (CP). Traditionally, these links are implemented by optical fibers or high-speed Ethernet, with each RRH having a dedicated link
to the CP. 
However, the large numbers of RRHs and the difficulty to reach some RRHs with wired connections make the dedicated links not always possible.
To this end, millimeter wave (mmWave) wireless fronthaul~\cite{dehos2014millimeter,stephen2017joint,hu2017joint} emerges as a cost-effective technique to enable flexible implementation of fronthaul links, which can operate on the largely unused mmWave bandwidth with highly directional antennas.
Meanwhile, in order to reduce the emission of $CO_2$ and build a more environmentally friendly communication system, energy harvested from renewable sources~\cite{zhou2015greendelivery,qin2017fronthaul,chen2017sustainable} such as solar and wind powers, can be exploited by RRHs as an alternative for traditional on-grid power supplies.


Recently, simultaneous wireless information and power transfer (SWIPT) has attracted great attention in the literature, which enables users to either decode data messages or harvest radio-frequency (RF) energy from the broadcast wireless signals.
Moreover, it is beneficial to integrate multiple antenna technologies~\cite{shi2014group,dai2014sparse,luo2015downlink,xiang2013coordinated}, especially multi-user MIMO into a Cloud-RAN for efficient information and energy transmissions~\cite{chen2018optimal,Boshkovska2017max,ariffin2017sparse,ng2015secure}. 
Thus, both data receivers (DRs) and energy receivers (ERs) in the network can be satisfied simultaneously by joint transmit beamforming.
In~\cite{chen2018optimal} and~\cite{Boshkovska2017max}, joint beamforming design for SWIPT is investigated without considering the capacity limitation of fronthaul links.
Specifically, throughput-energy trade-off regions for a sustainable Cloud-RAN are derived in~\cite{chen2018optimal}, and max-min fair beamforming design for energy transfer is studied in~\cite{Boshkovska2017max} under imperfect channel state information (CSI).
On the other hand, in both~\cite{ariffin2017sparse} and \cite{ng2015secure}, with given limited fronthaul capacities, total network transmit power is minimized for joint beamforming of SWIPT, where each DR and each ER are satisfied with a constant signal-to-interference-plus-noise (SINR) target and received RF energy target, respectively.
However, they only consider total energy minimization for a constant SINR target. How to improve the data service rate with providing fairness for all DRs, i.e., to achieve the maximum overall minimum data rate among all the DRs with optimal joint beamforming under per-RRH energy budget is still unknown in a green Cloud-RAN with limited fronthaul capacity.




In this paper, we consider a Cloud-RAN system, where the RRHs are all supplied with green energy and connected to the CP using mmWave wireless fronthaul links with limited capacity. 
In order to achieve a balanced user experience, joint transmit beamforming vectors will be optimized to maximize the minimum data rate among all the DRs, while each ER will be satisfied with sufficient RF energy at the same time.
Hence, an optimization problem is formulated to design the optimal joint beamforming vector, which is originally a non-convex problem with the fronthaul constraint in an $l_0$-norm form.
In order to handle the $l_0$-norm constraints, we approximate it by iteratively using the reweighted $l_1$-norm.
Although it is still non-convex due to the the nonlinear objective, it can be converted into an equivalent inverse problem, with which semidefinite relaxation (SDR) and bisection search can be applied to obtain the optimal minimum transmission rate and the corresponding joint transmit beamforming vector. 


The rest of this paper is organized as follows. In Section \ref{sec.model}, system model and problem formulation are introduced. Then, the proposed beamforming design for max-min fair SWIPT is presented in Section \ref{sec.beam_design}. Numerical simulation results are provided in Section \ref{sec.numerical}. Finally, Section \ref{sec.conclusion} concludes the paper. 

\section{System Model and Problem Formulation}\label{sec.model}

As shown in Fig. \ref{fig.cran}, a green Cloud-RAN system consists of $L$ RRHs with $M$ antennas, $K$ single-antenna DRs and $J$ single-antenna ERs. Each RRH $l \in \mathcal{L} = \{1,\ldots,L\}$ is powered by some renewable energy such as solar or wind power, and is connected to the BBU pool via a wireless fronthaul link of capacity $C_l$. 
Considering that the fronthaul links operate on mmWave frequencies with directional antennas, the interference between different wireless fronthaul links is negligible.
Meanwhile, all the ERs and DRs are served simultaneously in the downlink on the same frequency band. Specifically, each DR $k \in \mathcal{K} = \{1,\ldots,K\}$ is served by a network-wide beamforming vector $\mathbf{w}_k = [\mathbf{w}_{k1}^T,\ldots,\mathbf{w}_{kL}^T]^T \in \mathbb{C}^{ML\times 1}$, where $\mathbf{w}_{kl} \in \mathbb{C}^{M \times 1}$ is the beamforming vector at RRH $l$ for DR $k$. 
Similarly, each ER $j \in \mathcal{J} = \{1,\ldots,J\}$ is served by the beamforming vector $\mathbf{v}_j = [\mathbf{v}_{j1}^T,\ldots,\mathbf{v}_{jL}^T]^T \in \mathbb{C}^{ML\times 1}$.
The data symbol for DR $k$ and the energy symbol for ER $j$ are denoted by $s_k^\mathrm{D}$ and $s_j^{\mathrm{E}}$, respectively. 
Without loss of generality, we assume that the symbols are all independent with each other, which satisfy $\mathbb{E}[|s_k^\mathrm{D}|^2] = \mathbb{E}[|s_j^\mathrm{E}|^2] = 1$ for any $k \in \mathcal{K}$ and $j \in \mathcal{J}$. 
Notice that due to the energy symbol $s_j^\mathrm{E}$ is randomly generated, which carries no information but only satisfies the RF regulations. 
Hence, the received signal at each DR $k \in \mathcal{K}$ can be represented by 
\begin{align}
y_{k} = \mathbf{h}_k^H \mathbf{w}_k s_k^\mathrm{D} + \sum_{i=1,i\neq k}^{K} \mathbf{h}_k^H \mathbf{w}_i s_i^\mathrm{D} + \sum_{j=1}^J \mathbf{h}_k^H \mathbf{v}_j s_j^\mathrm{E} + n_k,
\end{align}
where $\mathbf{h}_k = [\mathbf{h}_{1k}^T,\ldots,\mathbf{h}_{Lk}^T]^T \in \mathbb{C}^{ML \times 1}$. Here, $\mathbf{h}_{lk} \in \mathbb{C}^{M \times 1}$ denotes the quasi-static complex channel vector from RRH $l$ to DR $k$, and $n_k$ is the additive white circularly symmetric complex Gaussian (CSCG) noise with identical variance $\sigma^2$ for each DR $k$.
Therefore, the SINR of DR $k$ is written by 
\begin{align}\label{eq.sinr}
\mathrm{SINR}_k = \frac{|\mathbf{h}_k^H \mathbf{w}_k|^2}{\sum\limits_{i=1,i\neq k}^{K} |\mathbf{h}_k^H \mathbf{w}_i|^2 + \sum\limits_{i=1}^J |\mathbf{h}_k^H \mathbf{v}_i|^2 + \sigma^2},
\end{align}
and thus the data rate at DR $k$ can be given by
\begin{align} \label{eq.rate_def}
R_k = \log(1 + \mathrm{SINR}_k).
\end{align}

\begin{figure}
\centering
\includegraphics[height = 6cm]{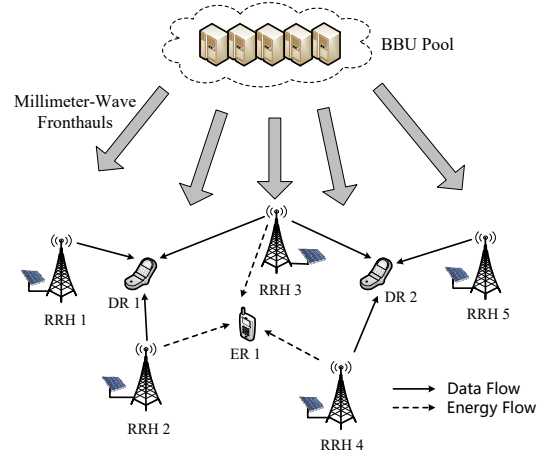}
\caption{A green Cloud-RAN system with wireless fronthaul links.}
\label{fig.cran}
\vspace{-1.5em}
\end{figure}

On the other hand, all the data symbols and energy symbols will be harvested by the ERs as RF energy. Thus, for each ER $j \in \mathcal{J}$, the harvested RF energy is proportional to the total received wireless signal power, which is given by
\begin{align}\label{eq.energy}
Q_j = \eta \left(\sum_{i=1}^K |\mathbf{g}_j^H \mathbf{w}_i|^2 + \sum_{i=1}^J |\mathbf{g}_j^H \mathbf{v}_i|^2\right),
\end{align}
where $\eta \in (0,1)$ is the RF energy conversion efficiency and $\mathbf{g}_{j} = [\mathbf{g}_{1j}^T,\ldots,\mathbf{g}_{Lj}^T]^T \in\mathbb{C}^{ML \times 1}$. Here, $\mathbf{g}_{lj} \in \mathbb{C}^{M \times 1}$ denotes the quasi-static complex channel vector from RRH $l$ to ER $j$.

In our model, we assume that the BBU pool can access global CSI of all DRs and ERs, based on which the sparse beamforming vectors $\{\mathbf{w}_k\}_{k=1}^K$ and $\{\mathbf{v}_j\}_{j=1}^J$ will be designed.
Due to the limited capacity of wireless fronthaul links, only a small group of RRHs will be selected to serve each DR $k$. 
If $\|\mathbf{w}_{kl}\|_2^2 \neq 0$, the data message for DR $k$ and the beamforming vector $\mathbf{w}_{kl}$ will be transmitted to RRH $l$.
If $\|\mathbf{w}_{kl}\|_2^2 = 0$, RRH $l$ is not associated with DR $k$.
For slow-varying channels, we only consider fronthaul consumption for data sharing, while the bandwidth required for CSI sharing and beamforming vector delivering can be ignored~\cite{dai2014sparse}. 
As a result, the total fronthaul bandwidth consumption of RRH $l$ can be written by $ \sum_{k=1}^K \left\| \|\mathbf{w}_{kl}\|_2^2 \right\|_0 \cdot R_k$, where the $l_0$-norm $\| \|\mathbf{w}_{kl}\|_2^2 \|_0$ denotes the association between DR $k$ and RRH $l$. 

Since each RRH is powered by renewable sources, we let $E_l$ denote green energy generated per second at RRH $l \in \mathcal{L}$. 
Notice that $E_l$ may not be equal for different RRHs, considering the spatial diversity of RRH deployment in different locations.
Moreover, the coherence time of wireless channel is much shorter than that of the renewable energy harvesting process at RRHs. I.e., the energy generation rate changes relatively slowly than CSI. 
Thus, the generated green energy $E_l$ at each RRH $l$ is assumed to be a pre-known constant~\cite{ng2015secure}. 

In summary, in order to maximize the minimum data rate among all the DRs, while guaranteeing each ER to be supplied with sufficient RF energy $Q_{\min}$, the downlink beamforming vector design for the above Cloud-RAN system can be formulated by an optimization problem as follows,
\begin{align}
\mathrm{(P1):} 
\max_{\{\mathbf{w}_{k}\}, \{\mathbf{v}_j\} } & \min_{k\in\mathcal{K}} \hspace{6 pt} {R}_k \nonumber
 \\
\mathrm{s.t.} \hspace{1em} & Q_j \geq Q_{\min},  \forall j \in \mathcal{J}, \label{eq.RF_energy_constraint} \\
& \sum_{k=1}^K \left\| \|\mathbf{w}_{kl}\|_2^2 \right\|_0 \cdot R_k \leq C_l, \forall l \in \mathcal{L}, \label{eq.fronthaul_constraint}\\
& \sum_{k=1}^K \|\mathbf{w}_{kl}\|_2^2 + \sum_{j=1}^J \|\mathbf{v}_{jl}\|_2^2 \leq E_l, \forall l \in \mathcal{L}, \label{eq.EH_energy_constraint}
\end{align}
where \eqref{eq.RF_energy_constraint} guarantees that the RF energy harvested by each ER is not lower than the RF energy target $Q_{\min}$, \eqref{eq.fronthaul_constraint} holds because the total fronthaul bandwidth consumption is limited by the link capacity $C_l$ at each RRH $l$. 
Moreover, \eqref{eq.EH_energy_constraint} represents the total transmission power at each RRH $l$ is constrained by the generated green energy $E_l$.

\begin{remark}
(P1) is a non-convex optimization problem because of the objective function and the constraints in \eqref{eq.RF_energy_constraint} and \eqref{eq.fronthaul_constraint}.
Particularly, the $l_0$-norm of the fronthaul capacity constraints in \eqref{eq.fronthaul_constraint} makes this problem even challenging to solve. We will show later that it can be approximated by using the reweighted $l_1$-norm.  
\end{remark}

\section{Beamforming Design for Max-Min Fair SWIPT}\label{sec.beam_design}

To design beamforming vectors for max-min fair SWIPT, a feasibility analysis will be firstly conducted to obtain the maximum target RF energy, and then the optimal joint beamforming design to maximize the minimum data rate of all DRs will be presented.

\subsection{Feasibility Analysis}
Due to the requirement of RF energy harvesting in \eqref{eq.RF_energy_constraint}, (P1) may not be always feasible, which makes it necessary to verify the feasibility of the target RF energy constraint $Q_{\min}$.
Thus, we have the following problem 
\begin{align}
\mathrm{(P2):} \max_{\{\mathbf{v}_j\}}  \hspace{4pt} & \min_{j\in\mathcal{J}} \hspace{6 pt}  \eta\sum_{i=1}^J |\mathbf{g}_j^H \mathbf{v}_i|^2  \nonumber
 \\
\mathrm{s.t.} \hspace{0.5em} & \sum_{j=1}^J \|\mathbf{v}_{jl}\|^2 \leq E_l, \forall l \in \mathcal{L}, \label{eq.EH_energy_constraint_Q}
\end{align}
where only energy beamforming is considered. As a result, all the beamforming vectors for data transmission as well as the wireless fronthaul link capacity constraints are removed from (P1).
Although (P2) is still a non-convex optimization problem, SDR can be applied to obtain the following problem, 
\begin{align}
\mathrm{(P3):} \max_{\{\mathbf{V}_j \succeq 0 \}}  \hspace{3pt} & \min_{j\in\mathcal{J}} \hspace{6 pt} \eta \sum_{i=1}^J \trace(\mathbf{G}_j\mathbf{V}_i) \nonumber
 \\
\mathrm{s.t.} \hspace{0.5em} & \sum_{j=1}^J \trace(\mathbf{V}_j\mathbf{A}_l) \leq E_l, \forall l \in \mathcal{L}, \label{eq.EH_energy_constraint_Q_matrix}
\end{align}
where we define $\mathbf{G}_j = \mathbf{g}_j \mathbf{g}_j^H$, $\mathbf{V}_j = \mathbf{v}_j \mathbf{v}_j^H$ and the block diagonal matrices $\mathbf{A}_l$ are defined as 
\begin{align}
\mathbf{A}_l = \diag{(\underbrace{0,\ldots,0}_{(l-1)M},\underbrace{1,\ldots,1}_{M},\underbrace{0,\ldots,0}_{(L-l)M})},\forall l \in \mathcal{L}.
\end{align}
It is worth noting that the rank-one constraint is relaxed for the energy beamforming covariance matrices $\{\mathbf{V}_j\}_{j \in \mathcal{J}}$.
Since point-wise minimum preserves concavity, (P3) is a convex optimization problem where the strong duality holds, which can be then efficiently solved by the interior point method~\cite{boyd2004convex}.
Furthermore, it can be proved that the optimal solution satisfies $\rank(\mathbf{V}_j) \leq 1$ for all $j \in \mathcal{J}$, which closely follows the proof in~\cite[Proposition 3.1]{xu2014multiuser} and will be omitted here due to page limitation.
By solving (P2), we now obtain the maximum RF energy target, i.e., the maximum value we can set for $Q_{\min}$.


\subsection{Optimal Beamforming Design}

Now we can consider (P1) under a feasible RF energy target $Q_{\min}$. As remarked after (P1), due to the $l_0$-norm fronthaul link capacity constraints in \eqref{eq.fronthaul_constraint}, it is challenging to obtain the global optimal solution to (P1). Thus, we will focus on algorithms to derive the local optimum of (P1).
Inspired by the approximation of $l_0$-norm using a convex reweighted $l_1$-norm widely adopted in compressive sensing~\cite{candes2008enhancing}, the total fronthaul bandwidth consumption can be written by
\begin{align}
\sum_{k=1}^K \left\|  \|\mathbf{w}_{kl}\|_2^2 \right\|_0\cdot R_k & \approx \sum_{k=1}^K \beta_{kl} \left\| \|\mathbf{w}_{kl}\|^2_2 \right\|_1 \cdot R_k, \\
 & = \sum_{k=1}^K \beta_{kl} \|\mathbf{w}_{kl}\|^2_2 \cdot R_k.
\end{align}
According to~\cite{dai2014sparse}, (P1) can be effectively solved with proper weights $\beta_{kl}$. 
To this end, the weights can be updated iteratively using the following formula,
\begin{align} \label{eq.beta_def}
\beta_{kl} = \frac{1}{\|\mathbf{w}_{kl}\|^2_2 + \tau}, \forall k \in \mathcal{K}, \forall l \in \mathcal{L},
\end{align}
where $\tau > 0$ is a small constant regularization factor and $\mathbf{w}_{kl}$ is the corresponding beamforming vector derived in the last round iteration. 


Even with the above approximation, the constraints in \eqref{eq.fronthaul_constraint} is still difficult to handle because of the non-convex term $R_k$. 
To address this, we propose to solve (P1) iteratively with $\beta_{kl}$ and $\hat{R}_k$ updated from last round iteration. 
In this way, by denoting $\mathbf{H}_k =  \mathbf{h}_k \mathbf{h}_k^H$, $\mathbf{W}_k =  \mathbf{w}_k \mathbf{w}_k^H$, (P1) can be reformulated as the following relaxed problem,
\begin{align}
(\mathrm{P4}):\hspace{2em} & \nonumber \\
\hspace{-0.8em} \max_{\{\mathbf{W}_{k}\}, \{\mathbf{V}_j\}}\hspace{1.2em} & \hspace{-1.2em}\min_{k \in \mathcal{K}} \frac{\trace(\mathbf{H_k}\mathbf{W}_k)}{ \sum\limits_{i=1,i\neq k}^K \trace(\mathbf{H_k} \mathbf{W}_i) + \sum\limits_{i=1}^J \trace(\mathbf{H}_k \mathbf{V}_i) + \sigma^2} \nonumber
 \\
\mathrm{s.t.} \hspace{2.5em} 
& \hspace{-1.5em} \sum_{i=1}^K \trace(\mathbf{G}_j \mathbf{W}_i) + \sum_{i=1}^J \trace(\mathbf{G}_j \mathbf{V}_i) \geq Q_{\min} / \eta, \forall j \in \mathcal{J}, \label{eq.RF_energy_constraint_matrix} \\
& \hspace{-1.5em} \sum_{k=1}^K \beta_{kl}\trace(\mathbf{W}_{k}\mathbf{A}_l)\hat{R}_k \leq C_l, \forall l \in \mathcal{L}, \label{eq.fronthaul_constraint_fixed}\\
& \hspace{-1.5em} \sum_{k=1}^K \trace(\mathbf{W}_{k}\mathbf{A}_l) + \sum_{j=1}^J \trace(\mathbf{V}_{j}\mathbf{A}_l) \leq E_l, \forall l \in \mathcal{L}, \label{eq.EH_energy_constraint_single_Matrix} \\
& \hspace{-1.5em} \mathbf{W}_{k} \succeq 0, \mathbf{V}_j \succeq 0, \forall k\in\mathcal{K}, \forall j\in\mathcal{J}, \label{eq.matrix_semidefinite}
\end{align}
where the rank-one constraints for all beamforming covariance matrices $\{\mathbf{W}_k\}_{k\in\mathcal{K}}$ and $\{\mathbf{V}_j\}_{j\in \mathcal{J}}$ are relaxed. 
However, it can be proved later that the optimal covariance matrices for (P4) are all rank-one. 
Notice that (P4) is still a non-convex optimization problem due to the objective function. 
Nevertheless, we can associate (P4) with its inverse problem, which can be represented by a weighted peak power minimization problem for all the RRHs as follows, 
\begin{align}
{(\mathrm{P5}):} \min_{\{\mathbf{W}_{k}\}, \{\mathbf{V}_j\}} \hspace{3em} & \hspace{-3em} \max_{l \in \mathcal{L}} \frac{\sum\limits_{k=1}^K \trace(\mathbf{W}_{k}\mathbf{A}_l) + \sum\limits_{j=1}^J \trace(\mathbf{V}_{j}\mathbf{A}_l)}{E_l} \nonumber
 \\
\mathrm{s.t.} \hspace{4.5em} 
& \hspace{-3em} \frac{1}{\gamma}\trace(\mathbf{H_k}\mathbf{W}_k) - \sum_{i=1,i\neq k}^K \trace(\mathbf{H_k} \mathbf{W}_i) \nonumber \\
& \hspace{-2.5em} - \sum_{i=1}^J \trace(\mathbf{H}_k \mathbf{V}_i) - \sigma^2 \geq 0, \forall k \in \mathcal{K}, \label{eq.SINR_constraint_matrix} \\
& \hspace{-3em} \eqref{eq.RF_energy_constraint_matrix}, \hspace{0.5em} \eqref{eq.fronthaul_constraint_fixed} \hspace{0.5em} \mathrm{and} \hspace{0.5em} \eqref{eq.matrix_semidefinite}. \nonumber
\end{align}
where we set a common SINR target $\gamma$ for all the DRs, with retaining the other constraints in (P4) except for the transmission power constraints in \eqref{eq.EH_energy_constraint_single_Matrix}.

In order to solve (P5), it can be finally reformulated into the following equivalent form, 
\begin{align}
{(\mathrm{P6}):}\min_{\{\mathbf{W}_{k}\}, \{\mathbf{V}_j\},\rho} \hspace{1em} &  \rho \nonumber
 \\
\mathrm{s.t.} \hspace{6em} 
& \hspace{-5em} \sum_{k=1}^K \trace(\mathbf{W}_{k}\mathbf{A}_l) + \sum_{j=1}^J \trace(\mathbf{V}_{j}\mathbf{A}_l)\leq \rho \cdot E_l, \forall l \in \mathcal{L}, \label{eq.energy_max_green_cons}  \\ 
& \hspace{-5em} \eqref{eq.RF_energy_constraint_matrix}, \hspace{0.5em} \eqref{eq.fronthaul_constraint_fixed}, \hspace{0.5em} \eqref{eq.matrix_semidefinite}\hspace{0.5em} \mathrm{and} \hspace{0.5em} \eqref{eq.SINR_constraint_matrix}, \nonumber
\end{align}
which is a convex optimization problem that can be efficiently solved by the interior method. 
Thus, the optimal solution of (P5) can be obtained from (P6). 
It can be easily verified that the optimal value of (P6) is a non-decreasing function of $\gamma$.
Moreover, the optimal beamforming covariance matrices $\{\mathbf{W}^*_k\}_{k\in \mathcal{K}}$ and $\{\mathbf{V}^*_j\}_{j\in \mathcal{J}}$ can be proved to be rank-one. 

\begin{lemma}\label{lm.rank_one}
If receiver channels are independently distributed, the optimal solution to (P6) satisfies $\rank(\mathbf{W}^*_k) \leq 1$, $\forall k \in \mathcal{K}$, and $\rank(\mathbf{V}^*_j) \leq 1$, $\forall j \in \mathcal{J}$, with probability one.
\end{lemma}

\begin{proof}
Please refer to Appendix \ref{app.rank_one}.
\end{proof}

\begin{remark}
From Lemma \ref{lm.rank_one}, we know that the optimal solution of (P5) is also rank-one. Thus, the optimal transmit covariance matrices $\{\mathbf{W}^*_k\}_{k\in \mathcal{K}}$ and $\{\mathbf{V}^*_j\}_{j\in \mathcal{J}}$ can be decomposed into vectors $\{\mathbf{w}^*_k\}_{k\in\mathcal{K}}$ and $\{\mathbf{v}^*_j\}_{j\in\mathcal{J}}$, respectively.
\end{remark}

In order to solve (P4), it will be connected with (P5) in the following lemma. 
To start with, it is worth noting that the optimal value of (P4) represents the maximum common SINR $\gamma_{\max}$ for all the DRs.
On the other hand, for a common SINR target $\gamma$, the optimal value of (P5) stands for the minimum weighted peak power consumption denoted as $h(\gamma)$.
For the sake of convenience, we define that a common SINR target $\gamma$ is achievable once it satisfies $\gamma \leq \gamma_{\max}$. 
In this way, these two problems can be connected in the following lemma.

\begin{lemma}\label{lm.sinr_achievable}
The common SINR target $\gamma$ is achievable if and only if it satisfies $h(\gamma) \leq 1$. 
\end{lemma}

\floatstyle{spaceruled}
\restylefloat{algorithm}
\begin{algorithm}[t] 
\caption{Bisection Search for $\gamma_{\max}$}\label{ag.1}
\begin{algorithmic}[1]
\State Set the initial upper and lower bounds for $\gamma_{\max}$ as $\gamma_{L} = 0$ and $\gamma_{U} = \max\limits_{k\in\mathcal{K}} \frac{\left(\sum\limits_{l=1}^L \sqrt{E_l}\|\mathbf{h}_{lk}\|_2 \right)^2}{\sigma^2}$.
\State Set $\gamma = \frac{\gamma_{L} + \gamma_{U}}{2}$ and then solve (P5); 
\While{$| h(\gamma) - 1| > \epsilon$}
\If{$h(\gamma) > 1$}
\State Update $\gamma_{L} = \gamma$;
\Else
\State Update $\gamma_{U} = \gamma$;
\EndIf
\State Update $\gamma = \frac{\gamma_{L} + \gamma_{U}}{2}$ and then solve (P5); 
\EndWhile
\State Return the optimal value $\gamma_{\max} = \gamma$ and the corresponding beamforming vectors $\{\mathbf{w}^*_k\}_{k\in\mathcal{K}}$ and $\{\mathbf{v}^*_j\}_{j\in\mathcal{J}}$ 
 by decomposing $\{\mathbf{W}^*_k\}_{k\in \mathcal{K}}$ and $\{\mathbf{V}^*_j\}_{j\in \mathcal{J}}$. 
\end{algorithmic}
\end{algorithm}

\floatstyle{spaceruled}
\restylefloat{algorithm}
\begin{algorithm}[t] 
\caption{Max-Min SINR Beamforming Design for (P1)}\label{ag.2}
\begin{algorithmic}[1]
\State Set the initial value for $\beta^{(0)}_{kl}$, $\hat{R}^{(0)}_{k}$ for $\forall k \in \mathcal{K}$, $\forall l \in \mathcal{L}$.
\State Set $n = 0$;
\While{$|\beta^{(n)}_{kl} - \beta^{(n-1)}_{kl}| > \epsilon_1$ or $|\hat{R}^{(n)}_k - \hat{R}^{(n-1)}_k| > \epsilon_2$}
\State Fixing $\beta^{(n)}_{kl}$, $\hat{R}^{(n)}_{k}$, solve (P4) to obtain the optimal value $\gamma^{(n)}_{\max}$ and the corresponding beamforming vectors $\{\mathbf{w}^*_k\}_{k\in\mathcal{K}}$ and $\{\mathbf{v}^*_j\}_{j\in\mathcal{J}}$;
\State Update $n = n+1$, $\beta^{(n)}_{kl} = \frac{1}{\trace(\mathbf{W}_{k}\mathbf{A}_l) + \tau}$ and $\hat{R}^{(n)}_k = \log\Bigg(1+\frac{\trace(\mathbf{H_k}\mathbf{W}_k)}{ \sum\limits_{i=1,i\neq k}^K \trace(\mathbf{H_k} \mathbf{W}_i) + \sum\limits_{i=1}^J \trace(\mathbf{H}_k \mathbf{V}_i) + \sigma^2}\Bigg)$,
\EndWhile
\State Return the optimal SINR $\gamma^* = \gamma^{(n)}_{\max}$ and the corresponding beamforming vectors $\{\mathbf{w}^*_k\}_{k\in\mathcal{K}}$ and $\{\mathbf{v}^*_j\}_{j\in\mathcal{J}}$.
\end{algorithmic}
\end{algorithm}

\begin{proof}
Firstly, it is straightforward to show that the common SINR target $\gamma$ is achievable when $h(\gamma) \leq 1$ holds.
For a given $\gamma$, $h(\gamma) \leq 1$ means $\frac{\sum\limits_{k=1}^K \trace(\mathbf{W}^*_{k}\mathbf{A}_l) + \sum\limits_{j=1}^J \trace(\mathbf{V}^*_{j}\mathbf{A}_l)}{E_l} \leq 1$ holds for all $l\in \mathcal{L}.$
Then, applying the optimal covariances $\{\mathbf{W}^*_k\}_{k\in \mathcal{K}}$ and $\{\mathbf{V}^*_j\}_{j\in \mathcal{J}}$ of (P5) to (P4), it can be easily verified that all the constraints in (P4) hold and thus we can know from \eqref{eq.SINR_constraint_matrix} that the common SINR target $\gamma$ satisfies $\gamma \leq \gamma_{\max}$.

On the other hand, we prove the necessity by contradiction. Suppose there exists an achievable $\gamma^\prime$ such that $h(\gamma^\prime) > 1$, where the transmit covariances to achieve such $\gamma^\prime$ in (P4) are denoted by $\{\mathbf{W}^\prime_k\}_{k\in \mathcal{K}}$ and $\{\mathbf{V}^\prime_j\}_{j\in \mathcal{J}}$.
Thus, applying the same covariance matrices in (P5), it can be verified that all the constraints in (P5) can be satisfied and a lower optimal peak power consumption $h(\gamma^\prime) \leq 1$ can be obtained, which contradicts with the assumption. Therefore, an achievable $\gamma^\prime$ will guarantee $h(\gamma^\prime) \leq 1$, which completes the proof. 
\end{proof}

\begin{remark}
According to Lemma \ref{lm.sinr_achievable} and the monotonicity of $h(\gamma)$, we know that the optimal value of (P4), i.e., $\gamma_{\max}$, satisfies $h(\gamma_{\max}) = 1$. Moreover, following the same optimal covariance matrices $\{\mathbf{W}^*_k\}_{k\in \mathcal{K}}$ and $\{\mathbf{V}^*_j\}_{j\in \mathcal{J}}$ obtained from (P5), the optimal solution to (P4) will be also rank-one.
In this way, for fixed factors $\beta_{kl}$ and $\hat{R}_{k}$, (P4) can be solved and $\gamma_{\max}$ can be obtained by one-dimension bisection search over $\gamma$, which is summarized in Algorithm \ref{ag.1}. 
Consequently, for the original (P1), the optimal max-min SINR beamforming vectors $\{\mathbf{w}^*_k\}_{k\in\mathcal{K}}$ and $\{\mathbf{v}^*_j\}_{j\in\mathcal{J}}$ can be obtained by iteratively solving (P4) using updated factors $\beta_{kl}$ and $\hat{R}_{k}$ according to \eqref{eq.beta_def} and \eqref{eq.rate_def}, which is also summarized in Algorithm \ref{ag.2}.
\end{remark}

\section{Numerical Analysis}\label{sec.numerical}
In this section, the proposed joint beamforming algorithm will be validated by numerical simulations and compared with other separate beamforming strategies. 
The network topology is shown in Fig. \ref{fig.system}, where there are $L = 3$ RRHs, $K = 6$ DRs and $J = 3$ ERs randomly deployed in the Cloud-RAN system. 
Notice that each RRH is equipped with $M = 2$ antennas. 
The channel power gain is modeled as $10^{-3}a/{d^{\alpha}}$, where $d$ is the distance in meters, $\alpha$ is the path-loss exponent set as $\alpha = 3$ and $a \sim \exp(1)$ is the Rayleigh fading. For all simulations, the results are averaged by 100 channel realizations.
Besides, we assume that the system bandwidth is $1$ MHz and the additive white Gaussian noise at the data receiver has a power spectral density $N_0 = 10^{-15}$ W/Hz.
For each RRH $l\in\mathcal{L}$, the generated green energy is assumed to be equal, i.e., $E_l = E_L$. 
Moreover, the capacity limit for the wireless fronthaul link is also equal, i.e., $C_l = C_L$.
Besides, the energy conversion efficiency factor for RF energy harvesting is $\eta = 50\%$.

\begin{figure}
\centering
\includegraphics[height = 7cm]{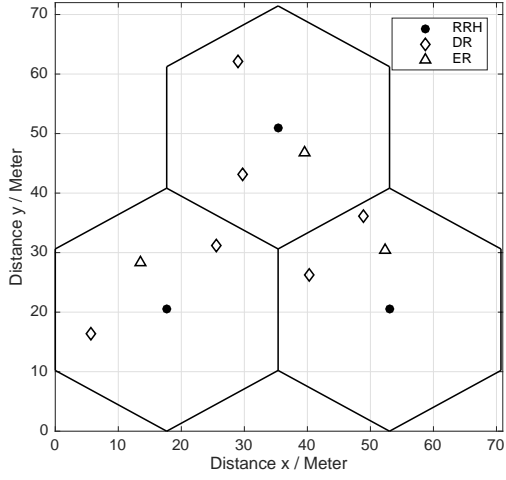}
\caption{Topology of the simulated Cloud-RAN system.}
\label{fig.system}
\end{figure}

\begin{figure}
\centering
\includegraphics[height = 7cm]{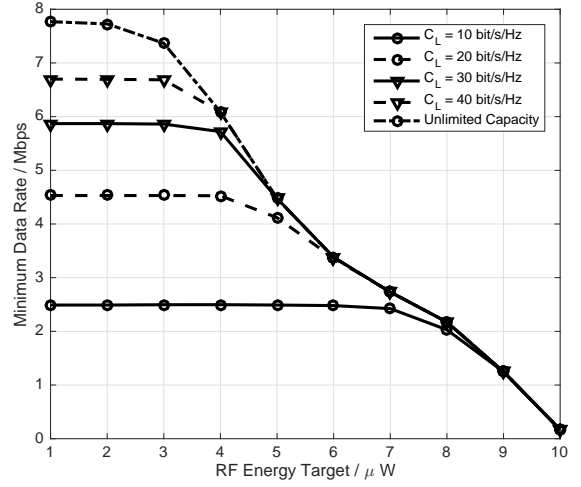}
\caption{Given different fronthaul capacity $C_L$, performance of max-min data rate versus RF energy target $Q_{\min}$ for generated green energy $E_L = 5$W.}
\label{fig.region_rate_q}
\vspace{-1em}
\end{figure}

\begin{figure}
\centering
\includegraphics[height = 7cm]{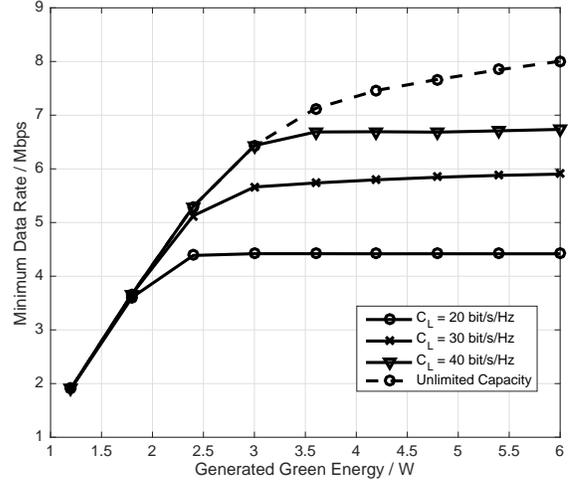}
\caption{Given different fronthaul capacity $C_L$, performance of max-min data rate versus the generated green energy $E_L$ for RF energy target $Q_{\min} = 1 \mu$W.}
\label{fig.region_rate_el_cl}
\vspace{-1em}
\end{figure}

In Fig. \ref{fig.region_rate_q}, performance of the max-min data rate versus the RF energy target $Q_{\min}$ is illustrated. Given the fronthaul capacity, it can be seen that the max-min data rate decreases as $Q_{\min}$ grows. Moreover, for the same $Q_{\min}$, the max-min data rate will become larger when given higher fronthaul capacity, which, however, approaches the ultimate max-min data rate for unlimited fronthaul capacity. 
Therefore, it can be inferred that the fronthaul capacity highly impacts the max-min rate for lower RF energy target, while a higher RF energy target dominates the max-min rate performance, regardless of the fronthaul capacity limit.

In Fig. \ref{fig.region_rate_el_cl}, performance of the max-min data rate versus the generated green energy $E_L$ is shown. Given the fronthaul capacity, it can be observed that as the generated green energy $E_L$ grows, the max-min data rate will increase accordingly. When $E_L$ becomes sufficiently large, the max-min data rate will finally saturate at some upper bound, which is determined by the fronthaul capacity. It is worth noting that for a larger fronthaul capacity, this upper bound will be higher, which, however, will be bounded by the ultimate max-min data rate under unlimited fronthaul capacity. 

The average number of associated RRHs per DR versus the fronthaul capacity $C_L$ is presented in Fig. \ref{fig.rrh}. Recall that the total number of RRHs $L = 3$. It can be seen from Fig. \ref{fig.rrh} that due to the limitation of wireless fronthaul capacity links, each DR can be only served by a small group of RRHs.  Given the RF energy target $Q_{\min}$, the number of RRHs associated with each DR will increase as the fronthaul capacity grows. 
Moreover, for the same $C_L$, the group of associated RRHs will expand as $Q_{\min}$ becomes larger. In fact, the data rate of each DR will become smaller for a larger $Q_{\min}$. Thus, each RRH can serve more DRs with lower data rate. 


\begin{figure}
\centering
\includegraphics[height = 7cm]{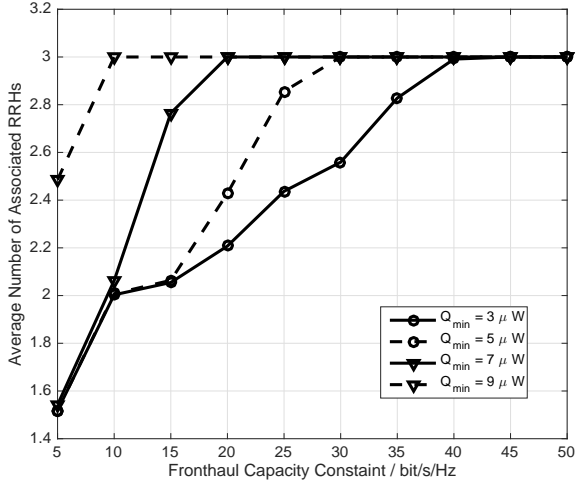}
\caption{Given different RF energy target $Q_{\min}$, average number of associated RRHs per DR versus fronthaul capacity $C_L$ for generated green energy $E_L = 5$W.}
\label{fig.rrh}
\vspace{-1em}
\end{figure}

To verify the performance of the proposed joint beamforming algorithm, a separated beamforming algorithm is introduced. Specifically, an energy beamforming vector will be firstly designed to satisfy each ER with sufficient RF energy. Then, data beamforming vectors will be optimized to maximize the minimum data rate among all the DRs.
Comparison of the max-min data rate for different beamforming strategies is shown in Fig.~\ref{fig.region_compare}. It can be observed that the proposed joint beamforming design outperforms the separate beamforming design for any RF energy target. Moreover, when a larger fronthaul capacity is given, the performance gap will become higher between these two strategies, from which we know that the proposed joint beamforming algorithm is superior to the separate beamforming algorithm.

\section{Conclusion}\label{sec.conclusion}

This paper studied joint transmit beamforming design to achieve max-min fair SWIPT in a green Cloud-RAN with mmWave wireless fronthaul.
In order to achieve a balanced user experience for separately located mobile users in the network, the minimum data rate among all the DRs has been maximized, while satisfying each ER with sufficient RF energy at the same time.
The formulated optimization problem is originally non-convex, which is challenging to solve, especially for the fronthaul capacity constraint in an $l_0$-norm form.
Thus, we have proposed a two-step iterative algorithm, which firstly approximates the $l_0$-norm constraint by the reweighted $l_1$-norm, and then derives the optimal max-min data rate and the corresponding joint beamforming vector using SDR and bi-section search. 
Numerical simulations demonstrates the superiority of the proposed joint beamforming algorithm to the separate beamforming algorithm. 
In our future work, joint beamforming will be designed to support SWIPT in a large-scale Cloud-RAN with massive MIMO and imperfect CSI.

\begin{figure}
\centering
\includegraphics[height = 7cm]{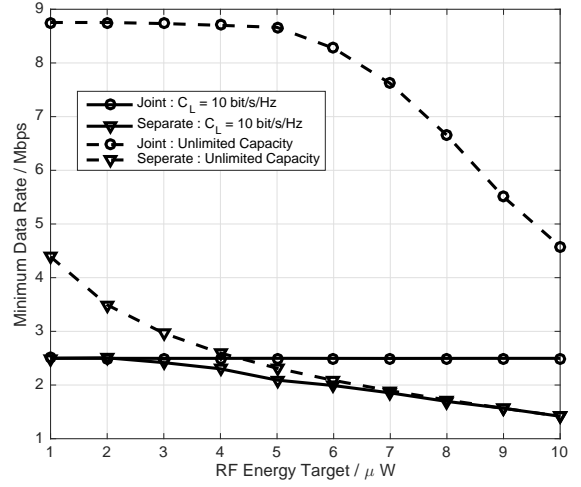}
\caption{Comparison of max-min data rate for different beamforming strategies for generated green energy $E_L = 10$W.}
\label{fig.region_compare}
\vspace{-1em}
\end{figure}

\appendix

\subsection{Proof of Lemma \ref{lm.rank_one}}\label{app.rank_one}
\begin{proof}
Introducing dual variables $\{\lambda_l \geq 0\}_{l\in\mathcal{L}}$, $\{\mu_k \geq 0\}_{k\in\mathcal{K}}$, $\{\nu_j \geq 0\}_{j\in\mathcal{J}}$, $\{\xi_l \geq 0\}_{l\in\mathcal{L}}$, $\{\mathbf{X}_k \succeq 0\}_{k\in\mathcal{K}}$ and $\{\mathbf{Y}_j \succeq 0 \}_{j\in\mathcal{J}}$, the Lagrangian function of (P6) can be written by
\begin{align}
& \mathcal{L}(\{\mathbf{W}_k\},\{\mathbf{V}_j\},\{\lambda_l\},\{\mu_k\},\{\nu_j\},\{\xi_l\},\{\mathbf{X}_k\},\{\mathbf{Y}_j\}) \nonumber \\
= & \rho + \sum_{l=1}^L \lambda_l \left( \sum_{k=1}^K \trace(\mathbf{W}_{k}\mathbf{A}_l) + \sum_{j=1}^J \trace(\mathbf{V}_{j}\mathbf{A}_l) - \rho E_l \right)  \nonumber \\
- &\!\sum_{k=1}^K\!\mu_k\!\left(\!\frac{1}{\gamma}\trace(\mathbf{H_k}\mathbf{W}_k)\!-\!\!\!\sum_{i=1,i\neq k}^K\!\!\!\!\trace(\mathbf{H_k}\!\mathbf{W}_i)\!-\!\sum_{i=1}^J\!\trace(\mathbf{H}_k \mathbf{V}_i)\!-\!\sigma^2\!\right) \nonumber \\
- & \sum_{j=1}^J \nu_j \left( \sum_{i=1}^K \trace(\mathbf{G}_j \mathbf{W}_i) + \sum_{i=1}^J \trace(\mathbf{G}_j \mathbf{V}_i) - Q_{\min} / \eta \right) \nonumber \\
+ & \sum_{l=1}^L \xi_l \left( \sum_{k=1}^K \beta_{kl}\trace(\mathbf{W}_{k}\mathbf{A}_l)\hat{R}_k - C_l \right) \nonumber \\
 - & \sum_{k=1}^K \trace(\mathbf{W}_k\mathbf{X}_k) - \sum_{j=1}^J \trace(\mathbf{V}_j \mathbf{Y}_j), \\
= & \sum_{k=1}^K \trace(\mathbf{B}_k \mathbf{W}_k) - \sum_{k=1}^K \trace\left(\mathbf{W}_k\left( \frac{\mu_k}{\gamma}\mathbf{H}_k + \mathbf{X}_k\right)\right) \nonumber \\
 & + \sum_{j=1}^J \trace(\mathbf{D}_j\mathbf{V}_j) - \sum_{j=1}^J \trace\left(\mathbf{V}_j \left( \nu_j \mathbf{G}_j + \mathbf{Y}_j \right)\right) + \Delta,
\end{align}
where we denote 
\begin{align}
& \Delta = \rho - \sum_{l=1}^L \lambda_l \rho E_l + \sum_{k=1}^K \mu_k \sigma^2 + \sum_{j=1}^J \nu_j Q_{\min} / \eta - \sum_{l=1}^L \xi_l C_l, \nonumber \\
& \mathbf{B}_k = \sum_{l=1}^L \left( \lambda_l + \xi_l\beta_{kl}\hat{R}_k \right) \mathbf{A}_l + \sum_{i=1,i\neq k}^K\mu_i\mathbf{H}_i - \sum_{j=1}^J \nu_j\mathbf{G}_j, \nonumber \\
& \mathbf{D}_j = \sum_{l=1}^L \lambda_l \mathbf{A}_l + \sum_{k=1}^K \mu_k\mathbf{H}_k - \sum_{i=1,i\neq j}^J \nu_i \mathbf{G}_i. \nonumber
\end{align}

Since (P6) is a convex optimization problem, the Slater's condition can be satisfied and then strong duality holds. 
Thus, by denoting $\mathbf{\Theta} = (\{\lambda_l\},\{\mu_k\},\{\nu_j\},\{\xi_l\})$, the dual problem can be written by
\begin{align}
\max_{\mathbf{\Theta},\{\mathbf{X}_k\},\{\mathbf{Y}_j\}} \min_{\{\mathbf{W}_k\},\{\mathbf{V}_j\}} \mathcal{L}(\{\mathbf{W}_k\},\{\mathbf{V}_j\},\mathbf{\Theta},\{\mathbf{X}_k\},\{\mathbf{Y}_j\}). \nonumber
\end{align}

Suppose that the optimal solution of the dual problem is $\mathbf{\Theta}^*$, $\mathbf{X}^*_k$ and $\mathbf{Y}^*_j$. Then, we have the following KKT conditions:
\begin{align}
\mathbf{W}^*_k \mathbf{X}^*_k & = 0, & \forall k \in \mathcal{K},\label{eq.wx_kkt} \\
\mathbf{V}^*_j \mathbf{Y}^*_j & = 0, & \forall j \in \mathcal{J}, \\
\mathbf{B}^*_k - \left(\frac{\mu_k^*}{\gamma}\mathbf{H}_k + \mathbf{X}^*_k\right) & = 0, & \forall k \in \mathcal{K}, \label{eq.bk_kkt} \\ 
\mathbf{D}^*_j - \left(\nu^*_j\mathbf{G}_j + \mathbf{Y}^*_j\right)  & = 0, & \forall j \in \mathcal{J}, 
\end{align} 
where $\mathbf{B}^*_k$ and $\mathbf{D}^*_j$ can be obtained by substituting the optimal dual variables into their expressions, respectively. 
Now, to prove $\rank(\mathbf{W}^*_k) = 1$, $\forall k \in \mathcal{K}$ with probability one, we will firstly show each $\mathbf{B}^*_k$ is positive definite by contradiction. 
Suppose that $\mathbf{B}^*_{k_0}, k_0 \in \mathcal{K}$ is a non-positive definite matrix. Thus, the beamforming matrix can be chosen as $\mathbf{W}_{k_0} = \kappa\mathbf{w}_{k_0}\mathbf{w}_{{k_0}}^H$, where $\kappa > 0$ is a scaling factor and $\mathbf{w}_{k_0}$ is the eigenvector corresponding to one of the non-positive eigenvalues of $\mathbf{B}^*_{k_0}$. As a result, the optimal value of (P6) can be obtained by
\begin{align}
& \min_{\{\mathbf{W}_k\}} \mathcal{L}(\{\mathbf{W}_k\},\{\mathbf{V}^*_j\},\mathbf{\Theta^*},\{\mathbf{X}^*_k\},\{\mathbf{Y}^*_j\}) \\
& = \Delta^* + \kappa \mathbf{w}_{k_0}^H\mathbf{B}^*_{k_0} \mathbf{w}_{k_0} - \kappa \mathbf{w}_{k_0}^H\left( \frac{\mu^*_{k_0}}{\gamma}\mathbf{H}_{k_0} + \mathbf{X}^*_{k_0}\right)\mathbf{w}_{k_0} \nonumber \\
& + \sum_{k=1,k\neq{k_0}}^K \trace(\mathbf{B}^*_k \mathbf{W}_k) - \sum_{k=1,k\neq{k_0}}^K \trace\left(\mathbf{W}_k\left( \frac{\mu^*_k}{\gamma}\mathbf{H}_k + \mathbf{X}^*_k\right)\right) \nonumber \\
& + \sum_{j=1}^J \trace(\mathbf{D}^*_j\mathbf{V}^*_j) - \sum_{j=1}^J \trace\left(\mathbf{V}^*_j \left( \nu_j \mathbf{G}_j + \mathbf{Y}^*_j \right)\right),
\end{align}
where $\kappa\mathbf{w}_{k_0}^H\mathbf{B}^*_{k_0}\mathbf{w}_{k_0}$ and 
$-\kappa\mathbf{w}_{k_0}^H\left( \frac{\mu^*_{k_0}}{\gamma}\mathbf{H}_{k_0}+\mathbf{X}^*_{k_0}\right)\mathbf{w}_{k_0}$ are both non-positive, which leads to an unbounded optimal value when $\kappa \rightarrow \infty$. However, it contradicts with the fact that the optimal value of (P6) is non-negative, thus strong duality does not hold. Therefore, each $\mathbf{B}^*_{k}$ is positive definite with probability one and $\rank(\mathbf{B}^*_{k}) = WL$, since the channel vectors $\mathbf{h}_{k}$ and $\mathbf{g}_{j}$ are independently distributed. Then, according to \eqref{eq.bk_kkt}, we have
\begin{align}
\rank(\mathbf{B}^*_{k}) \leq \rank\left(\frac{\mu_k^*}{\gamma}\mathbf{H}_k\right) + \rank(\mathbf{X}^*_k),
\end{align}
which indicates that 
\begin{align}
\rank(\mathbf{X}^*_k) \geq \rank(\mathbf{B}^*_{k}) - \rank\left(\frac{\mu_k^*}{\gamma}\mathbf{H}_k\right) \geq ML - 1.
\end{align}
As a result, with the KKT condition in \eqref{eq.wx_kkt}, we know that 
\begin{align}
\rank(\mathbf{W}^*_{k}) \leq ML - \rank(\mathbf{X}^*_k) = 1.
\end{align}
Following similar steps, it can be proved that $\rank(\mathbf{V}^*_{k}) \leq 1$ holds with probability one. This completes the proof of Lemma \ref{lm.rank_one}.
\end{proof}





\end{document}